\documentclass[a4paper, 11pt]{article}
\usepackage{fullpage}

\usepackage{t1enc}
\usepackage[utf8]{inputenc}

\usepackage{latexsym,amssymb,amsmath,amsthm,graphicx,enumerate,mathtools}
\usepackage{cite,comment,easy-todo}
\usepackage{hyperref}

\newcommand{\Rset}{\ensuremath{\mathbb R}}

\newcommand{\cI}{\ensuremath{\mathcal I}}

\theoremstyle{plain}

\newtheorem{thm}{Theorem}
\newtheorem{theorem}[thm]{Theorem}

\newtheorem{proposition}[thm]{Proposition}
\newtheorem{corollary}[thm]{Corollary}

\newtheorem{lemma}[thm]{Lemma}

\newtheorem{claim}[thm]{Claim}

\theoremstyle{definition}

\newtheorem{remark}[thm]{Remark}
\newtheorem{example}[thm]{Example}

\usepackage{braket}
\newcommand{\vote}{{\sf vote}}
\newcommand{\weakvote}{{\sf vote}^{\bullet}}
\newcommand{\score}{{\sf score}}

\title{Solving the Maximum Popular Matching
Problem \\with Matroid Constraints\thanks{A preliminary version of this paper has appeared in the proceedings of the 12th Japanese-Hungarian Symposium on Discrete Mathematics and Its Applications, 2023.}}
\author{Gergely Csáji\thanks{MTA-ELTE Momentum Matroid Optimization Research Group, Department of Operations Research, Eötvös Loránd University, Budapest, Hungary. Email: \texttt{csajigergely@student.elte.hu} and Mechanism Design Research Group, Institute of Economics and Regional Studies, Hungary.  Email: \texttt{csaji.gergely@krtk.hu}}\and
Tamás Király\thanks{ELKH-ELTE Egerv\'ary Research Group, Department of Operations Research, E\"otv\"os Loránd University, Budapest, Hungary. Email: \texttt{tamas.kiraly@ttk.elte.hu}}\and
Yu Yokoi\thanks{Department of Mathematical and Computing Science,
School of Computing,
Tokyo Institute of Technology, Tokyo, Japan. Email: \texttt{yokoi@c.titech.ac.jp}}}
\date{\empty}

\begin{document}
\maketitle

\begin{abstract}
We consider the problem of finding a maximum popular matching in a many-to-many matching setting with two-sided preferences and matroid constraints. This problem was proposed by Kamiyama~(2020) and solved in the special case where matroids are base orderable. Utilizing a newly shown matroid exchange property, we show that the problem is tractable for arbitrary matroids. We further investigate a different notion of popularity, where the agents vote with respect to lexicographic preferences, and show that both existence and verification problems become coNP-hard even in the $b$-matching case.
\end{abstract}

\section{Introduction}

The notion of \emph{popular matching} is a natural adaptation of the notion of weak Condorcet winner \cite{nicolas1785essai} to the marriage model of Gale and Shapley \cite{GS62}, where agents of a two-sided market have strict preference orders on admissible agents on the other side. A matching is called popular if it does not lose a head-to-head election against any other matching, i.e., it is a Condorcet winner in the election among matchings. It is a well-known fact (sometimes called the Condercet paradox) that a weak Condorcet winner does not always exist in the general voting setting. Remarkably, existence is guaranteed in the marriage model: G\"ardenfors \cite{gardenfors1975match} showed that every stable matching is popular. In fact, stable matchings are the smallest popular matchings, so the notion of popular matching can be considered as a relaxation of stable matching, where we sacrifice pairwise stability in order to achieve larger size.

Several years after the results of G\"ardenfors, popular matchings came into the focus again in the 2000s due to their interesting algorithmic properties. Huang and Kavitha \cite{huang2011popular} showed that a maximum size popular matching in the marriage model can be found in polynomial time. In contrast, recently it was shown by Faenza et al.~\cite{faenza2019popular} and Gupta et al.~\cite{gupta2019popular,gupta2021popular} (simultaneously and independently) that deciding the existence of a popular matching in the roommates (i.e., non-bipartite) model is NP-complete.

Just as in the case of the stable marriage problem, the results have been extended to many-to-many matchings. The concept of Condorcet winner is not so straightforward in this setting, because there are several different ways in which an agent can compare two matchings based on the sets of partners. Nonetheless, remarkable findings by Brandl and Kavitha \cite{brandl2018popular,brandl2019two} show that popular many-to-many matchings exist under a rather restrictive definition of popularity, and furthermore, the largest such matching has the maximum size even among matchings satisfying a much less restrictive notion of popularity.


Nasre and Rawat \cite{nasre2017popularity} introduced a many-to-many model where agents can have classifications in their preference lists, and classes can have upper quotas. Kamiyama \cite{kamiyama2020popular} proposed a far-reaching generalization of this model, by extending the laminar nested classification of Nasre and Rawat to a model involving arbitrary matroids. He gave an algorithm that returns a popular matching, based on Fleiner's algorithm for finding a matroid kernel \cite{fleiner2001matroid,fleiner2003fixed}, which is a generalization of the notion of stable matching to matroid intersection. 
For the maximum size popular matching problem, however, Kamiyama could only prove the correctness of his algorithm for the special case of weakly base orderable matroids. This includes the problem of Nasre and Rawat, but does not include graphic matroids for example. He left open the question whether there is a polynomial-time algorithm for arbitrary matroids. 

In this paper, we give an affirmative answer to this question. We show that \emph{the maximum popular matching problem with two-sided preferences and arbitrary matroid constraints can be solved in polynomial time}, by essentially the same algorithm as in \cite{kamiyama2020popular}. The key tool in extending the proof from weakly base orderable matroids to arbitrary matroids is a new matroid exchange property that can be formulated in terms of voting between two independent sets in a matroid with a given linear order on the elements (an \emph{ordered matroid}). This exchange property will be described in Theorem \ref{thm:nonpos}, and we will prove it by combining matroid techniques with the duality between weighted bipartite matching and minimum cover.
Since this exchange property may potentially be useful in other areas of matroid theory too, we include an equivalent form of the result here that can be easily understood in purely matroid theoretic terms.

\begin{theorem}\label{thm:maxmin}
   Let $M=(S, \cI)$ be a matroid, let $\succ$ be a linear order on $S$, and let $A$ and $B$ be disjoint bases of $M$. For each $ab \in A \times B$, let $w(ab)=1$ if $a\prec b$ and $w(ab)=0$ if $a \succ b$. Let $E_A=\{ab\in A \times B:\, A-a+b \in \cI\}$ and $E_B=\{ab\in A \times B:\, B+a-b \in \cI\}$. Then,
   \[\max\{w(N):\, \text{$N$ is a perfect matching in $E_A$}\} \geq \min\{w(N'):\, \text{$N'$ is a perfect matching in $E_B$}\}.\]
\end{theorem}

We note that the special case when $A$ is an optimal base with respect to the ordering $\succ$ was proved previously in \cite{csaji2022approximation} by the present authors.


The precise definitions for various notions of popularity will be given in Section \ref{sec:popularity}; here we only give a high-level idea. 
We present our results in the framework of matroid intersection, which is equivalent to Kamiyama's model, but allows us to better describe the difference between the more restrictive and less restrictive popularity notions. It is also closer to the original matroid kernel problem defined by Fleiner \cite{fleiner2001matroid}. In our framework, two ordered matroids are given on the same ground set $S$, both as direct sums: $M_1= M^1_1 \oplus M^1_2 \oplus \dots \oplus M^1_{k_1}$ and $M_2= M^2_1 \oplus M^2_2 \oplus \dots \oplus M^2_{k_2}$. Agents correspond to the summands in the direct sums, so there are $k_1+k_2$ agents, each corresponding to a matroid $M^i_j=(S^i_j, \cI^i_j)$. 

To get a pairwise comparison of two common independent sets $X$ and $Y$, the agent corresponding to $M^i_j$ compares $X\cap S^i_j$ and $Y\cap S^i_j$ based on the ordering of the elements of $S^i_j$. The details of this are described in Section \ref{sec:popularity}; here we just note that an agent may cast multiple votes based on how the elements of  
$(X\setminus Y)\cap S^i_j$ and $(Y\setminus X) \cap S^i_j$ can be paired to each other. The votes of all agents are added to obtain the total vote between $X$ and $Y$.

The common independent set $X$ is popular if there is no common independent set $Y$ that gets more votes in such a pairwise comparison. We show that a maximum size popular common independent set can be found efficiently. Furthermore, we prove a property similar to the one by Brandl and Kavitha mentioned above: there always exists a common independent set satisfying a remarkably restrictive definition of popularity that has the maximum size among all common independent sets satisfying weaker popularity properties.  

We also investigate another notion of popularity, called lexicographic popularity. Here, each agent has only one vote, and the agents compare common independent sets in a lexicographic way. 
Lexicographic preferences have been of considerable interest recently, as they arise in many applications. Cechlárová et al. \cite{cechlarova2014pareto} studied Pareto-optimal matchings in the many-to-many matching problem with lexicographic one-sided preferences. Biró and Csáji \cite{biro2022strong} investigated the strong core and Pareto optimality with two-sided lexicographic preferences. Closest to our work is the paper of Paluch \cite{paluch2014popular}, which studied popular and clan-popular matchings in the many-to-one matching problem with one-sided lexicographic preferences. 
We show that, in contrast to the previous notion of popularity, a lexicographically popular common independent set does not always exist and both the search and verification questions regarding lexicographic popularity are coNP-hard, even in the restricted case of $b$-matchings with constant degrees and capacities.

The rest of the paper is structured as follows. In Section \ref{sec:kernel}, we describe the matroid kernel problem, and the relationship between matroid kernels and two-sided stable matching with matroid constraints. 
In Section \ref{sec:popularity}, we define the various notions of voting and popularity that we consider in the popular matroid intersection problem, and we describe their relationship to the popularity notions used in the literature on many-to-many matchings. We also present our new result on matroid exchanges, which is proved in Section~\ref{sec:nonpos}. In Section~\ref{sec:algorithm}, we describe the algorithm for the maximum size popular matroid intersection problem and the proof of its correctness. Finally, in Section \ref{sec:lex} we define lexicographic popularity and provide hardness results for the related search and verification problems. 

\section{Ordered matroids and matroid kernels}\label{sec:kernel}
A {\em matroid} is a pair $(S, \cI)$ of a finite set $S$ and a nonempty family $\cI\subseteq 2^S$ satisfying the following two axioms: (i) $A\subseteq B\in \cI$ implies $A\in \cI$, and (ii) for any $A, B\in \cI$ with $|A|<|B|$, there is $v\in B\setminus A$ with $A+v\in \cI$. A set in $\cI$ is called an {\em independent set}, and an inclusion-wise maximal one is called a {\em base}. By axiom (ii), all bases have the same size, which is called the {\em rank} of the matroid.

A {\em circuit} of a matroid is an inclusionwise minimal dependent set. The \emph{fundamental circuit} of an element $x\in S\setminus B$ for a base $B$ is the unique circuit in $B+x$. We will use the following well-known property.

\begin{proposition}[Strong circuit axiom]
If $C,C'$ are circuits, $x\in C\setminus C'$, and $y\in C \cap C'$, then there is a circuit $C'' \subseteq C \cup C'$ such that $x \in C''$ and $y \notin C''$. \qed 
\end{proposition}

In our proofs in Section~\ref{sec:nonpos}, we will use the fact that matroids are closed under operations such as {\em direct sum}, {\em restriction}, {\em contraction}, and {\em truncation}. For these operations and other basics on matroids, we refer the reader to \cite{schrijver2003combinatorial}.

An \emph{ordered matroid} is a triple $(S,\cI,\succ)$ such that $(S, \cI)$ is a matroid and $\succ$ is a linear order on $S$. The linear order determines an optimal base in the following sense: for any weight vector $w\in \Rset^S$ which satisfies $w_x>w_y \Leftrightarrow x \succ y$, the unique maximum weight base is the same. We call this base $A$ the \emph{optimal base} of $(S,\cI,\succ)$; it is characterized by the property that $u \succ v$ for every $u \in A$ and $v\in S\setminus A$ for which $A-u+v \in \cI$. 



Let $M_1=(S,\cI_1,\succ_1 )$ and $M_2=(S,\cI_2,\succ_2 )$ be ordered matroids on the same ground set $S$, and let $I \in \cI_1 \cap \cI_2$ be a common independent set. We say that an element $v\in S\setminus I$ is {\em dominated} by $I$ in $M_i$ if $I+v \notin \cI_i$ and $u \succ_i v$ for every $u \in I$ for which $I-u+v \in \cI_i$. We call a common independent set $I \in \cI_1 \cap \cI_2$ an \emph{$(M_1,M_2)$-kernel} if every $v \in S \setminus I$ is dominated by $I$ in $M_1$ or $M_2$. If an element $v \in S \setminus I$ is dominated in neither $M_1$ nor $M_2$, we say that $v$ \emph{blocks} $I$.

It was shown by Fleiner \cite{fleiner2001matroid,fleiner2003fixed} that matroid kernels always exist and have
the same size -- in fact, they have the same span in both matroids. He also gave a matroidal version of the Gale--Shapley algorithm that finds an $(M_1,M_2)$-kernel efficiently, in $\mathcal{O}(|S|^2)$ time.

To understand the relation between our problem formulation and the formulation of Kamiyama \cite{kamiyama2020popular}, it is instructive to see the equivalence of the matroid kernel model above and the model of \emph{stable matchings with matroid constraints}, as described below. Let $G=(V_1,V_2;E)$ be a bipartite graph, and for each $v\in V_1 \cup V_2$, let $M_v=(\delta_G(v), \mathcal{I}_v, \succ_v)$ be an ordered matroid, where $\delta_G(v)$ denotes the set of edges incident to $v$. An edge set $I \subseteq E$ is called a \emph{matching} if $I\cap \delta_G(v) \in \mathcal{I}_v$ for every $v \in V_1 \cup V_2$. A matching $I$ is \emph{stable} if for any $e=v_1v_2\in E\setminus I$, either $I\cap \delta_G(v_1)$ is the optimal base of $M_{v_1}$ restricted to $(I+e)\cap \delta_G(v_1)$, or $I\cap \delta_G(v_2)$ is the optimal base of $M_{v_2}$ restricted to $(I+e)\cap \delta_G(v_2)$.

We show that this is actually equivalent to the matroid kernel model. To formulate stable matchings with matroid constraints as a matroid kernel problem, let $M_1$ be the matroid on ground set $E$ obtained as the direct sum of the matroids $M_{v}$ ($v\in V_1$), and let $\succ_1$ be obtained by arbitrarily extending the linear orders $\succ_v$ ($v\in V_1$) into a linear order on $E$. We define $M_2$ and $\succ_2$ similarly using $V_2$. It is easy to see that $(M_1,M_2)$-kernels are exactly the stable matchings. Conversely, a matroid kernel problem can be written as a stable matching problem with matroid constraints, where $G$ consists of two vertices and $|S|$ parallel edges between them.


\section{Voting and popularity in matroid intersection}\label{sec:popularity}

\subsection{Voting in ordered matroids}

For clarity of presentation, we use the word `pairing' instead of `matching' for a family of disjoint pairs of elements from two given disjoint subsets $A$ and $B$. Thus, a \emph{pairing between $A$ and $B$} is a matching in the complete bipartite graph with vertex classes $A$ and $B$, while a \emph{perfect pairing} is a perfect matching in the same graph. 

Consider an ordered matroid $M=(S, \cI, \succ)$, where $M$ is given as a direct sum $M=M_1 \oplus M_2 \oplus \dots \oplus M_k$ for some positive integer $k$ and matroids $M_j=(S_j,\cI_j)$ ($j \in [k]$). 
Given an ordered pair of independent sets $(I,J)$, let $N$ be a pairing between $I \setminus J$ and $J \setminus I$ and consider the following five conditions:
\begin{itemize}
\item[(1)] $I-u+v \in \cI$ for every $uv \in N$, where  $u\in I\setminus J$ and $v\in J\setminus I$.
\item[(2)] Any element $v\in J\setminus I$ with $I+v\not\in \cI$ is covered by $N$.
\item[(3)] Any element $u\in I\setminus J$ with $J+u\not\in \cI$ is covered by $N$.
\item[(4)] Every $uv\in N$ satisfies $u,v\in S_j$ for some $j \in [k]$.
\item[(5)] For every $j\in [k]$, the number of pairs of $N$ induced by $S_j$ is $\min\{|S_j \cap (I \setminus J)|, |S_j \cap (J \setminus I)|\}$. 

\end{itemize}
Intuitively, conditions (1), (4) and (5) mean that the agent corresponding to $M_j$ compares $I$ and $J$ by pairing the elements of $S_j \cap (I\setminus J)$ to elements of $S_j\cap (J \setminus I)$ with which they can be exchanged, and comparing each pair. Conditions (2) and (3) ensure that the pairing is reasonable when $S_j \cap (I\setminus J)$ and $S_j\cap (J \setminus I)$ have different sizes.

We say that $N$  is a {\em feasible pairing} for $(I,J)$ if (1)-(5) hold.
For two independent sets $I$ and $J$ and a feasible pairing $N$ for $(I,J)$, we define 
\[
    \vote(I,J,N)=|\{\,uv \in N: u\succ v\,\}|-|\{\,uv \in N: u\prec v\,\}|+|I|-|J|,
\label{eq:vote}
\]
where $u\in I\setminus J$ and $v\in J\setminus I$.
Considering the most adversarial feasible  pairing, we define
\begin{align*}
	\vote(I,J)&= \min\{\,\vote(I,J,N): N \text{~is a feasible pairing for $(I,J)$}\,\}.
\end{align*}

It turns out that the following property is crucial for proving the main results.

\begin{theorem}\label{thm:nonpos}
$\vote(I, J)+\vote(J,I)\leq 0$.
\end{theorem}

We present the proof in the next section.
Here, we make two remarks about the theorem.

\begin{remark}
The statement can be shown easily if the matroid $(S, \cI)$ is weakly base orderable. Indeed, the definition of weak base orderability implies the existence of a pairing $N$ that is feasible for both $(I, J)$ and $(J,I)$, from which we obtain $\vote(I, J)+\vote(I, J)\leq \vote(I, J, N)+\vote(J, I, N)=0$. For a general matroid, however, such a pairing $N$ may not exist, which makes it difficult to extend the proof arguments in previous works \cite{kamiyama2020popular,kavitha2014size} to general matroids. We use Theorem~\ref{thm:nonpos} to overcome this difficulty.
\end{remark}

\begin{remark}
It is easy to see that Theorem~\ref{thm:maxmin} corresponds to the case of Theorem~\ref{thm:nonpos} where $I=A$ and $J=B$ are disjoint bases. Indeed, in this case, the perfect matchings in $E_A$ are exactly the feasible pairings for $(A,B)$, while the perfect matchings in $E_B$ are exactly the feasible pairings for $(B,A)$. We will also see in the proof of Theorem~\ref{thm:nonpos} that the general case can be reduced to the case where $I$ and $J$ are disjoint bases, so the two theorems are actually equivalent.
\end{remark}

We also use a more restrictive notion of popularity, whose definition uses a broader class of pairings. The idea is that we can consider $M$ as a trivial direct sum $M=M_1$, i.e., we regard the $k$ agents as a single ``aggregate agent''. We also drop condition (3) from the definition of feasible pairing.
We say that $N$ is a {\em weakly feasible pairing} for $(I, J)$ if (1) and (2) hold.
For two independent sets $I$ and $J$ and a weakly feasible pairing $N$ for $(I,J)$, we can define $\vote(I,J,N)$ the same way as for feasible pairings. Considering the most adversarial weakly feasible  pairing,~we~define
\begin{align*}
	\weakvote(I,J)&= \min\{\,\vote(I,J,N): N \text{~is a weakly feasible pairing for $(I,J)$}\,\}.
\end{align*}



The following is an immediate corollary of Theorem \ref{thm:nonpos}.

\begin{corollary}\label{thm:popineq}
The following sequence of inequalities holds
for any pair of  independent sets $I$ and $J$:
$\weakvote(I,J) \leq \vote(I,J) \leq -\vote(J,I) \leq -\weakvote(J,I)$.   
\end{corollary}

\subsection{Popularity in matroid intersection}

Let $M_1=(S, \cI_1, \succ_1)$ and $M_2=(S,\cI_2, \succ_2)$ be ordered matroids, given as direct sums $M_1= M^1_1 \oplus M^1_2 \oplus \dots \oplus M^1_{k_1}$ and $M_2= M^2_1 \oplus M^2_2 \oplus \dots \oplus M^2_{k_2}$.
For an ordered pair $(I,J)$ of common independent sets and $i\in \{1,2\}$, we define $\vote_i(I,J)$ as $\vote(I,J)$ with respect to $M_i$. We call a common independent set $I\in \cI_1\cap \cI_2$ {\em popular} if $\vote_1(I,J)+\vote_2(I,J) \geq 0$ for every $J\in \cI_1\cap \cI_2$. Also, we call $I\in \cI_1\cap \cI_2$ {\em defendable} if $\vote_1(J,I)+\vote_2(J,I) \leq 0$ for every $J\in \cI_1\cap \cI_2$.

\begin{remark}
It is important to remember that feasible pairings for $(I,J)$ are not the same as feasible pairings for $(J,I)$. When considering popularity of $I$, we compare it to $J$ by taking a feasible pairing for $(I,J)$ that is worst possible for $I$. In contrast, defendability of $I$ is determined by considering a feasible pairing for $(J,I)$ that is best possible for $I$.  
\end{remark}

By using $\weakvote_i$ instead of $\vote_i$, we can define a stronger version of popularity and a weaker version of defendability, which we call {\em super popularity} and {\em weak defendability}, respectively. Note that these do not depend on the given decompositions of $M_1$ and $M_2$ into direct sums. The relation between these notions can be derived from Theorem \ref{thm:nonpos}.

\begin{corollary}\label{cor:implication}
The following implications hold for any $I\in \cI_1\cap \cI_2$:
\[ \text{$I$ is super popular} \Rightarrow \text{$I$ is popular} \Rightarrow \text{$I$ is defendable} \Rightarrow \text{$I$ is weakly defendable}
\]
\end{corollary}
\begin{proof}
It follows from Corollary \ref{thm:popineq} that
\begin{multline*}
\weakvote_1(I,J)+\weakvote_2(I,J) \leq \vote_1(I,J)+\vote_2(I,J) \\ \leq -\vote_1(J,I)-\vote_2(J,I) \leq -\weakvote_1(J,I)-\weakvote_2(J,I)
\end{multline*}
for any $J\in \cI_1\cap \cI_2$. This gives the required implications.
\end{proof}

In Section \ref{sec:nonpos}, we prove Theorem \ref{thm:nonpos}.
In Section \ref{sec:algorithm}, we show that an abstract version of Kamiyama's algorithm \cite{kamiyama2020popular} outputs a common independent set that is super popular, and largest among all weakly defendable common independent sets. This generalizes several results in previous works. In Kamiyama's model \cite{kamiyama2020popular}, feasible pairings are defined by 
conditions (1),(2),(4), and (5). Then, our result shows that the algorithm's output is a largest popular common independent set also in his definition. In the partition matroid case (i.e., $b$-matching case) studied by Brandl and Kavitha \cite{brandl2018popular}, our popularity notion coincides with their popularity and our defendability coincides with their weak popularity. Therefore, our result generalizes the result of Brandl and Kavitha \cite{brandl2018popular} that we can efficiently find a popular matching that is largest among all weakly popular matchings.


\section{Proof of Theorem \ref{thm:nonpos}}\label{sec:nonpos}

Recall that $M=(S, \cI, \succ)$ is an ordered matroid, given as a direct sum $M=M_1 \oplus M_2 \oplus \dots \oplus M_k$ for some positive integer $k$ and matroids $M_j=(S_j,\cI_j)$ ($j \in [k]$). 
 Let $I\in \cI$ and $J\in \cI$ be arbitrary independent sets. Our aim is to prove that $\vote(I,J)+\vote(J,I) \leq 0$.

For $j\in [k]$, let $I_j\coloneqq S_j\cap I$ and $J_j\coloneqq S_j\cap J$. 
If $|I_j|\leq |J_j|$, then let $A_j\subseteq J_j\setminus I_j$ be a set satisfying $I_j\cup A_j\in \cI_j$ and $|A_j|=|J_j|-|I_j|$, and set $I'_j\coloneqq I_j\setminus J_j$ and $J'_j\coloneqq J_j\setminus (I_j\cup A_j)$.
If $|I_j|>|J_j|$, then define $I'_j$ and $J'_j$ similarly by exchanging the roles of $I_j$ and $J_j$. 
In any case, we have $|I'_j|=|J'_j|=\min\{|S_j \cap (I \setminus J)|, |S_j \cap (J \setminus I)|\}$.
Let $M'_j$ be the matroid obtained by restricting $M_j$ to $I_j\cup J_j$, contracting $(I_j\cap J_j)\cup A_j$, and truncating to the size of $|I'_j|$.
The ground set of $M'_j$ is partitioned into two bases $I'_j$ and $J'_j$.
Let $M'=(S', \cI')$ be the direct sum $M'_1 \oplus \dots \oplus M'_{k}$. 
The ground set $S'$ of $M'$ is partitioned into two bases $I'\coloneqq I'_1\cup\cdots\cup I'_{k}$ and $J'\coloneqq J'_1\cup \cdots\cup J'_{k}$.

Let $G_I=(I',J';E_I)$ be the bipartite graph with 
$E_I=\{\,uv: u\in I', v\in J', I'-u+v\in \cI' \,\}$, and let
$G_J=(I',J';E_J)$ where $E_J=\{\,uv: u\in I', v\in J', J'+u-v\in \cI' \,\}$.
Since $I'$ and $J'$ are bases of $M'$, each of $G_I$ and $G_J$ admits a perfect matching.

\begin{claim}
Any perfect matching of $G_I$ is a feasible pairing for $(I,J)$, and any perfect matching of $G_J$ is a feasible pairing for $(J,I)$.
\end{claim}
\begin{proof}
By symmetry, it is enough to prove the first statement. Let $N$ be a perfect matching in $G_I$. By definition, $N$ is a pairing between $I \setminus J$ and $J \setminus I$. We show that $N$ satisfies conditions (1)-(5).

To see (1), consider $uv \in N$ such that $u \in I'_j$ and $v \in J'_j$. Then $uv \in E_I$ implies that $I'_j-u+v$ is independent in $M'_j$, so $(I\cap S_j)-u+v$ is independent in $M_j$ by the construction of $M'_j$. Since this holds for any $j\in [k]$, $I-u+v\in \cI$.

To show (2), consider an element $v\in J\setminus I$ not covered by $N$. Then $v\in A_j$ for some $j$ such that $|J_j|>|I_j|$. By definition, $A_j \cup I_j$ is independent in $M_j$, so $I_j+v\in \cI_j$. By the definition of direct sum, we then have $I+v\in \cI$. Condition (3) can be shown similarly, by exchanging the role of $I$ and $J$.

Conditions (4) and (5) are satisfied because $N$ gives a perfect pairing between  $I'_j$ and $J'_j$ for every $j \in [k]$.
\end{proof}

For $uv \in E_I$, let $w(uv)=1$ if $u\prec v$, and $w(uv)=0$ if $u \succ v$, and let $k$ be the maximum weight of a perfect matching in $E_I$. Then $\vote(I,J)\leq |I'|-2k$. 
By duality (between maximum weight perfect matching and minimum cover), 
there exists an integer function $\pi$ on $S'$
such that $\sum_{v\in S'}\pi(v)=k$ and
$\pi(u)+\pi(v)\geq w(uv)$ for every $uv \in E_I$.

We now consider the same weight function on $E_J$: let $w(uv)=1$ if $u\prec v$, and $w(uv)=0$ if $u \succ v$. 
Let $E$ consist of the edges $uv \in E_J$ which satisfy $\pi(u)+\pi(v)\geq w(uv)$.

\begin{lemma}\label{lem:perfect}
The bipartite graph $G=(I',J';E)$ has a perfect matching.
\end{lemma}
\begin{proof}
In the proof, we work with the matroid $M'$, so the term `circuit' refers to circuits of $M'$. Suppose for contradiction that there exists a subset $X$ of $I'$
such that the set of its neighbors in $G$, that we denote by $Y$, is
smaller than $X$. We introduce a new ordering $\succ'$ on the elements of $S'$:
$a \succ' b$ if either $\pi(a)<\pi(b)$, or $\pi(a)=\pi(b)$ and $a \succ b$ (we will only compare pairs
inside $I'$ or inside $J'$). The following claim is the main ingredient of the proof.

\begin{claim} 
Let $C$ be a circuit of $M'$ such that $C\cap I' \subseteq X$, and let $v$ be
the worst element of  $C \cap J'$ according to $\succ'$. Then $v \in Y$.
\end{claim}
\begin{proof} 
Suppose for contradiction that $v \notin Y$. There must exist a
vertex $u \in C\cap X$ such that $uv \in E_J$. Indeed, otherwise we could
eliminate the elements of $C\cap X$ one by one using the strong circuit
axiom with the fundamental circuits for $J'$, while retaining the property
that $v$ is in the circuit; in the end, we would obtain a circuit inside $J'$,
which is impossible. So, there is a vertex $u \in C\cap X$ such that $uv\in E_J$.
Let us call a vertex $v' \in J'$ \emph{bad} if either $\pi(u)+\pi(v')<0$, or
$\pi(u)+\pi(v')=0$ and $u \prec v'$. The vertex $v$ is bad because $uv \in E_J$ and $v
\notin Y$. Since $v$ is the  worst element of  $C \cap J'$ according to $\succ'$,
we have that every $v'\in C \cap J'$ is bad. Thus, $uv' \notin E_I$ holds for
every $v' \in C \cap J'$ (since $\pi(u)+\pi(v')\geq w(uv')$ if $uv' \in E_I$).
In other words, $u$ is not in the fundamental circuit of $v'$ for $I'$ for any
$v' \in C \cap J'$. But then we could eliminate the elements of $C\cap J'$
one by one using the strong circuit axiom with the fundamental circuits
for $I'$, while retaining the property that $u$ is in the circuit; in the end
we would obtain a circuit inside $I'$, which is impossible. This
contradiction proves the claim.
\end{proof}
We now show that $G$ has a perfect matching by getting a contradiction. For each $u \in X$, let $C_u$ be a circuit such that $C\cap
I' \subseteq X$, $u$ is the worst element of $C\cap X$ according to $\succ'$, and
subject to that, the worst element in $C\cap Y$ according to $\succ'$ is best
possible. Note that $C_u$ exists, because the fundamental circuit of $u$ for $J'$ is a candidate, and each candidate circuit has an element in $C\cap Y$ by the previous claim. 

Let $y(u)$ denote the worst element in $C_u\cap Y$ according to
$\succ'$. Since $|Y|<|X|$, there exist $u_1\in X$ and $u_2 \in X$ such that
$y(u_1)=y(u_2)$; we may assume $u_1\prec'u_2$. Let $y=y(u_1)=y(u_2)$; notice
that $y \in C_{u_1} \cap C_{u_2}$ and $u_1 \in C_{u_1} \setminus C_{u_2}$.
By the strong circuit axiom, we can obtain a circuit $C$ such that $C
\subseteq C_{u_1} \cup C_{u_2}-y$, and $u_1 \in C$. The existence of this
circuit contradicts the choice of $C_{u_1}$.
\end{proof}

To prove Theorem~\ref{thm:nonpos}, consider the perfect matching $N\subseteq E_J$ given by Lemma~\ref{lem:perfect}. Then $w(N)\leq \sum_{v
\in S'} \pi(v)=k$, so $N$ has at most $k$ edges $uv$ for which $u \prec v$.
This means that $\vote(J,I)\leq 2k-|I'|$, and therefore $\vote(I,J)+\vote(J,I) \leq 0$.
This completes the proof of the theorem.

\section{Algorithm}\label{sec:algorithm}
Here we describe Kamiyama's algorithm \cite{kamiyama2020popular} in a generalized form.
Given a pair of ordered matroids $M_i=(S, \cI_i, \succ_i)~(i\in \{1,2\})$, we construct an extended instance $M^*_i=(S^*, \cI^*_i, \succ^*_i)$ $(i\in \{1,2\})$ obtained by replacing each element with two parallel copies.
Let the extended ground set be $S^*=\cup_{u\in S}\{ x(u), y(u)\}$. The elements $x(u)$ and $y(u)$ are respectively called {\em $x$-copy} of $u$ and {\em $y$-copy} of $u$.
The independent set families are defined by
\[\cI_i^*=\{\,I^* \subseteq S^*: \pi(I^*)\in \cI_i,\  |I^*\cap\{x(u),y(u)\}|\leq 1\ (\forall u \in S)\,\},\]
where $\pi(I^*)=\{\,u\in S: I^*\cap \{x(u),y(u)\} \neq \emptyset \,\}$. 

The linear order $\succ^*_i$ on $S^*$ is defined as follows. 
In $\succ^*_1$, the $x$-copy of any element is preferred over the $y$-copy of any element, and the original preferences are preserved for the copies of the same type (e.g.,  $u\succ_1 v\Leftrightarrow x(u)\succ^*_1 x(v),\, y(u)\succ^*_1 y(v)$). In $\succ^*_2$, the roles of $x$ and $y$ are exchanged; the $y$-copies are preferred over the $x$-copies, and the original preferences are preserved for the copies of the same type.
Kamiyama's algorithm is described as follows:
\begin{enumerate}
	\item Find an $(M^*_1, M^*_2)$-kernel $I^*$.
	\item Output $I\coloneqq \pi(I^*)$.
\end{enumerate} 
Note that we can find a matroid kernel $I^*$ in the first step in $\mathcal{O}(|S|^2)$ by Fleiner's algorithm \cite{fleiner2001matroid, fleiner2003fixed}.

Let $I$ be the output of the algorithm. 
We show that $I$ is super popular and largest among all weakly defendable common independent sets,
which implies that $I$ is a maximum popular common independent set by Corollary~\ref{cor:implication}.
To this end, we provide the following lemma.

\begin{lemma}\label{lem:I}
For any $J\in \cI_1\cap \cI_2$ and any weakly feasible pairings $N_1$ and $N_2$ for $(I, J)$ with respect to matroids $M_1$ and $M_2$, respectively, we have $\vote_1(I, J, N_1)+\vote_2(I,J,N_2)\geq 0$.
Moreover, if $|J|>|I|$, then $\vote_1(I, J, N_1)+\vote_2(I,J,N_2)>0$.
\end{lemma}

Before providing the proof of this lemma, we show that it easily implies the following theorems, which are our main results.

\begin{theorem}
The output $I$ of the algorithm is super popular and is largest among all weakly defendable common independent sets.
\end{theorem}
\begin{proof}
The first claim of Lemma~\ref{lem:I} implies $\weakvote_1(I,J)+\weakvote_2(I,J)\geq 0$ for any $J\in \cI_1\cap \cI_2$, and hence $I$ is super popular. By Corollary~\ref{cor:implication}, then $I$ is weakly defendable.
By the second claim of Lemma~\ref{lem:I}, any common independent set $J\in \cI_1\cap \cI_2$ larger than $I$ satisfies $\weakvote_1(I,J)+\weakvote_2(I,J) > 0$, and hence $J$ is not weakly defendable. Thus, $I$ is a largest weakly defendable common independent set.
\end{proof}

Since we have Corollary~\ref{cor:implication} and the algorithm runs in polynomial time, the following theorem holds.

\begin{theorem}
Given two ordered matroids $M_1=(S, \cI_1, \succ_1)$ and $M_2=(S,\cI_2, \succ_2)$, one can find a maximum popular common independent set in polynomial time.
\end{theorem}

We now provide the proof of Lemma~\ref{lem:I}. It uses arguments similar to those used in Kavitha~\cite{kavitha2014size} and Kamiyama \cite{kamiyama2020popular}.
\begin{proof}[\bfseries{Proof of Lemma~\ref{lem:I}}]
Since each $N_i$ is a weakly feasible pairing, $I-u+v\in \cI_i$ for any $uv\in N_i$ and $I+v\in \cI_i$ for any $v\in J\setminus I$ not covered by $N_i$. 
By the stability of $I^*$, any element in $J\setminus I$ is covered by $N_1$ or $N_2$.
Consider the bipartite graph $G=(I\setminus J, J\setminus I; N_1\cup N_2)$, which is decomposed into alternating paths, cycles, and isolated vertices in $I\setminus J$.
For each path/cycle $P$, define its score as
\begin{align*}
\score(P)=&+|\{\,uv\in P :\,  uv\in N_i,\, u\succ_i v \text{ for some } i\in \{1,2\}\,\}| \\
&- |\{\,uv\in P :\,  uv\in N_i,\, u\prec_i v \text{ for some } i\in \{1,2\}\,\}|\\
&+2(|P\cap (I\setminus J)|-|P\cap (J\setminus I)|),
\end{align*}
where we assume $u\in I\setminus J$ and $v\in I\setminus J$ and identify $P$ with its edge set (resp., its vertex set) in the first and second terms (resp., in the third term).
Note that $\vote_1(I, J, N_1)+\vote_2(I,J,N_2)$ equals the sum of the scores of all cycles/paths in $G$ plus $2\cdot\#\text{\{isolated vertices of $I\setminus J$ in $G$\}}$.
Therefore, showing $\score(P)\geq 0$ for any path/cycle $P$ completes the proof of the first claim of Lemma~\ref{lem:I}.

Let $u_0 v_1 u_1 v_2 u_2 \dots v_k u_k$ be the elements on $P$ appearing in this order where $u_\ell\in I\setminus J$ and $v_\ell\in J\setminus I$ for each $\ell$, 
and we set $u_0=\emptyset$ if $P$ starts at $J\setminus I$, we set $u_k=\emptyset$ if $P$ ends at $J\setminus I$, and let $u_0=u_k$ if $P$ is a cycle.
Without loss of generality, we assume $u_{\ell-1} v_\ell\in N_1$ and $u_{\ell} v_{\ell}\in N_2$ for each $\ell$.

Consider the triple $u_{\ell-1} v_{\ell} u_{\ell}$ for $\ell=1,2,\dots,k$.
Since $I^*$ is stable, each of $x(v_\ell)$ and $y(v_\ell)$ should be dominated by $I^*$ in $M^*_1$ or $M^*_2$.
Note that any $x$-copy (resp., $y$-copy) is preferred to any $y$-copy (resp., $x$-copy) in $\succ^*_1$ (resp., $\succ^*_2$) and that we have $u_{\ell-1} v_\ell\in N_1$ and $u_{\ell} v_{\ell}\in N_2$. Note also that $u_{\ell-1}=\emptyset$ (resp., $u_\ell=\emptyset$) implies that $v_\ell$ is uncovered in $N_1$ (resp., in $N_2$), and hence $I^*+y(v_\ell)\in \cI_1$ (resp., $I^*+x(v_\ell)\in \cI_2$). From these, we obtain the following conditions. Here, an element $u\in I\setminus J$ is called {\em $x$-type} (resp., {\em $y$-type}) if $I^*\cap \{x(u), y(u)\}=x(u)$ (resp., $y(u)$).
\begin{enumerate}
\setlength{\itemindent}{4mm}
\item[(a)] If $u_{\ell-1}$ and $u_{\ell}$ are both $x$-type, then $u_{\ell-1} \succ_1 v_{\ell}$ or $u_{\ell} \succ_2 v_{\ell}$.
\item[(b)] If $u_{\ell-1}$ and $u_{\ell}$ are both $y$-type, then $u_{\ell-1} \succ_1 v_{\ell}$ or $u_{\ell} \succ_2 v_{\ell}$.
\item[(c)] If $u_{\ell-1}$ and $u_{\ell}$ are $y$-type and $x$-type, respectively, then $u_{\ell-1} \succ_1 v_{\ell}$ and $u_{\ell} \succ_2 v_{\ell}$.
\item[(d)] If $u_{\ell-1}=\emptyset$, then $u_{\ell} \succ_2 v_{\ell}$ and $u_{\ell}$ is $y$-type.
\item[(e)] If $u_{\ell}=\emptyset$, then $u_{\ell-1} \succ_1 v_{\ell}$ and $u_{\ell-1}$ is $x$-type.
\end{enumerate}
The amount of votes obtained by the comparisons on $u_{\ell-1} v_\ell\in N_1$ and $u_{\ell} v_{\ell}\in N_2$ is nonnegative in all of the above cases, and in particular, it is $2$ in case (c).
This amount can be $-2$ only in the unlisted case, i.e., when $u_{\ell-1}$ and $u_{\ell}$ are $x$-type and $y$-type, respectively. 
Consider calculating the sum of the first two terms of $\score(P)$
by counting votes along $P$ from $u_0$ to $u_k$. 
The value increases by $2$ when $u_\ell$ turns from $y$-type to $x$-type, does not decrease when its type does not change, and decreases at most by $2$ when $u_\ell$ turns from $x$-type to $y$-type.
If $P$ is a cycle, we can immediately obtain $\score(P)\geq 0$.

We then assume that $P$ is a path. 
By the above arguments, the sum of the first two terms of $\score(P)$ is at least  $2\cdot(\#\{u_\ell \text{ turns from } y\text{-type to } x\text{-type}\}-\#\{u_\ell \text{ turns from } x\text{-type to } y\text{-type}\})$.
The third term of $\score(P)$, i.e., $2(|P\cap (I\setminus J)|-|P\cap (J\setminus I)|)$, is $-2$/$0$/$2$ if both/either/none of $u_0$ and $u_k$ is $\emptyset$. 
With the conditions (d) and (e), these imply $\score(P)\geq 0$.

\medskip
Finally, we prove the second claim of the lemma.
Suppose $|J|>|I|$. As we observed before, all elements in $J\setminus I$ are covered by $N_1\cup N_2$.
Since $|I\setminus J|<|J\setminus I|$, there exists a path $P=u_0 v_1 u_1 v_2 u_2 \dots v_k u_k$ in $G$ that starts and ends at $J\setminus I$, i.e., $u_0=u_k=\emptyset$. Then, the third term of $\score(P)$ is $-2$. 
By (d) and (e), we have $u_1\succ_2 v_1$ and $u_{k-1}\succ_1 v_k$, from which we obtain $2$ votes.
From (d) and (e),  we also obtain that $u_1$ is $y$-type while $u_{k-1}$ is $x$-type, and hence $\#\{u_\ell \text{ turns from } y\text{-type to } x\text{-type}\}$ is strictly larger than $\#\{u_\ell \text{ turns from } x\text{-type to } y\text{-type}\}$.
These imply that the sum of the first two terms of $\score(P)$ is at least $4$. Thus, $\score(P)\geq 2>0$, and hence $\vote_1(I, J, N_1)+\vote_2(I,J,N_2)>0$.
\end{proof}

\section{Lexicographic Preferences}\label{sec:lex}

In the previous sections, we showed that finding a maximum popular matching in two-sided markets can be done in polynomial time, even if the two sides have arbitrary matroid constraints. 
However, our definition of popularity is not the only possible definition, and we may conceive other natural definitions of popularity for many-to-many settings.
In this section we take a different approach and define popularity with respect to a much simpler voting rule, where the agents compare the two matchings/independent sets lexicographically. This means that they care mostly about their best element being as good as possible and with regard to that, their second best element being as good as possible, etc.
This also implies that each agent has only one vote in the sense that they must choose a vote from the set $\{ -1,0,+1\}$ depending on which independent set they like better, similar to the one-to-one matching case. Note that in this model a smaller indepent set can be better for an agent than a much larger one, if the best element the agent obtains in the smaller one is better. 

To make our proofs easier to follow, in this section we consider only partition matroids (with arbitrary upper bounds on the partition classes). As all our results are hardness results, they naturally extend to arbitrary matroids. In the case of partition matroids, where the voters correspond to the partition classes of the two matroids, we can model the instance with a bipartite graph $G=(U,W;E)$, where the vertices of $U\cup W$ correspond to the agents who vote and the edge set corresponds to the common ground set of the two matroids. We assume that each agent $v\in U\cup W$ has a capacity $b(v)$ and a strict order $\succ^v$ over their adjacent edges, which we denote by $E(v)$. 

Next, we define lexicographic popularity formally, for the case of the above $b$-matching problem.
We call an edge set $\mu \subseteq E$ a \textit{$b$-matching} if $|\mu \cap E(v)|\le b(v)$ for all $v\in U\cup W$.
Given two $b$-matchings $\mu $ and $\mu '$ in $G$ and a vertex $v\in U\cup W$, we say that $\mu $ is \textit{lexicographically better than $\mu '$ for $v$}, denoted by $\mu \succ^v_{\rm lex} \mu '$, if $(E(v)\cap (\mu \cup \mu '))\setminus (\mu \cap \mu ')$ is non-empty and the best element of this set according to the order $\succ^v$ is in $\mu$. (In the context of arbitrary matroids, this would be generalized by saying that common independent set $I$ is lexicographically better than $I'$ for the agent corresponding to $S_j^i$, $i\in \{ 1,2\}$, if the best element of $ (S_j^i \cap (I\cup I')) \setminus (I\cap I')$ according to the order $\succ_i$ is in $I$).
Let $\vote_{\rm lex}^v(\mu ,\mu ') \in \{ -1,0,+1\}$ denote the vote of agent $v$ when comparing $\mu '$ to $\mu $, that is, it is $+1$ if $v$ lexicographically prefers $\mu $ to $ \mu '$, $-1$ if $v$ lexicographically prefers $\mu '$ to $\mu $, and 0 otherwise. (Note that a vote can only be zero if $E(v)\cap \mu =E(v)\cap \mu '$, i.e. $v$ obtains the same set of edges in both $b$-matchings). Let $\vote_{\rm lex}(\mu ,\mu ')=\sum_{v\in U\cup W}\vote_{\rm lex}^v(\mu ,\mu ')$ denote the sum of votes of all agents. We say that a $b$-matching $\mu $ is \textit{lexicographically popular}, if for any $b$-matching $\mu '\subseteq E$ it holds that $\vote_{\rm lex}(\mu ,\mu ')\ge 0$. Otherwise, if $\vote_{\rm lex}(\mu ,\mu ')<0$, we say that $\mu '$ \textit{dominates} $\mu$. Clearly, it holds that $\vote_{\rm lex}(\mu ,\mu ')=-\vote_{\rm lex}(\mu ',\mu )$ as exchangeability is symmetric for partition matroids.

In the standard one-to-one popular matching problem, lexicographic popularity and popularity coincide. However, when capacities can be larger, lexicographic popularity and popularity differ. 
First, we observe that in contrast to popular matchings, a lexicographically popular matching does not always exist.

\begin{figure}
    \centering
    \includegraphics{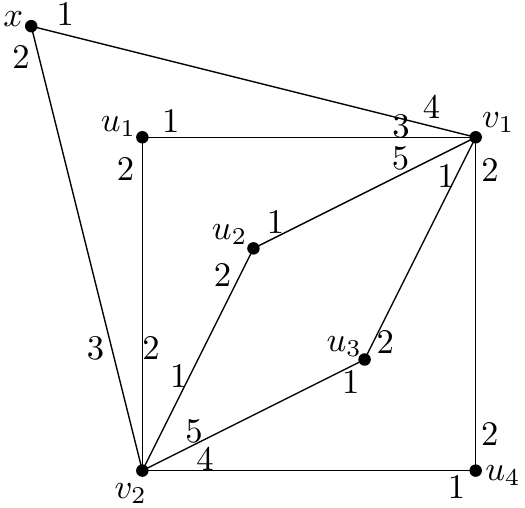}
    \caption{The instance given in Example \ref{ex:no-instance}. The numbers on the edges correspond to the preferences of the agents (smaller numbers are better).}
    \label{fig:no-instance}
\end{figure}

\begin{example}
\label{ex:no-instance}
We give an example with no popular matching with respect to lexicographic preferences. We will also use this example as a gadget later and exploit its properties. 

We have 7 agents, $x,u_1,\dots ,u_4$ on one side of the graph and $v_1,v_2$ on the other. The capacities of agents $u_1$ and $u_4$ are 1, the capacities of $u_2,u_3,v_1,v_2$ are 2 and the capacity of $x$ is some number $q\ge 1$. (We do not fix $q$, because later in the hardness proofs we will use that no popular matching exists for any $q\ge 1$ in this example). Next, we describe the preferences.
\begin{center}
\begin{tabular}{ll|ll}
$x:$ & $v_1\succ v_2$    &  $v_1:$ & $u_3\succ u_4\succ u_1\succ x\succ u_2$ \\
$u_1:$ & $v_1\succ v_2$     &  $v_2:$ & $u_2\succ u_1\succ x\succ u_4\succ u_3$ \\
$u_2:$ & $v_1\succ v_2$ & & \\
$u_3:$ & $v_2\succ v_1$ & &\\
$u_4: $ & $v_2\succ v_1 $ & & \\
\end{tabular}
\end{center}
The instance is illustrated in Figure \ref{fig:no-instance}.

Suppose that there is a lexicographically popular matching $\mu$. Then, the edges $(u_2,v_2)$ and $(u_3,v_1)$ must be in $\mu$. Indeed, if one of them, say $(u_2,v_2)$ is not in $\mu$, then $u_2$ is not saturated, so both $v_2$ and $u_2$ can improve by adding $(u_2,v_2)$ to $\mu$, and at most one agent gets worse (if $v_2$ is saturated and has to drop an edge). 

Agent $u_4$ must be matched in $\mu$, because otherwise $v_1$ is either unsaturated or has a worse partner. So, by adding $(u_4,v_1)$ to $\mu$ and deleting the worse edge of $v_1$ if $v_1$ was saturated, we obtain a matching where $v_1,u_4$ both improve, and at most one agent gets worse. If $u_4$ is matched to $v_2$, then one of $\{ x,u_1\}$ has to be totally unmatched in $\mu$. So, if agent $v_2$ drops $u_4$ and takes the free one of $\{ x,u_1\}$, then only $u_4$ gets worse and two agents improve, contradicting the lexicographic popularity of $\mu$. So we have that $(u_4,v_1)\in \mu$.

We can see that $u_1$ must be matched by replacing $u_4$ by $u_1$ and $v_1$ by $v_2$ in the above argument. As we have already seen that $v_1$ is saturated by $u_3$ and $u_4$, $(u_1,v_2)$ must be in $\mu$. We obtained that $\mu =\{ (u_1,v_2),(u_2,v_2),(u_3,v_1),(u_4,v_1)\}$. However, $\mu$ is dominated by the matching $\mu '=\{ (u_1,v_1),(u_2,v_1),(u_3,v_2),(u_4,v_2)\}$, because $u_1,u_2,u_3,u_4$ all improve, and only $v_1$ and $v_2$ get worse.

Hence, we have shown that no lexicographically popular matching exists in this instance if $x$ has capacity at least $1$. However, if we add $q$ dummy agents $d_1,d_2,\dots ,d_q$ who are only adjacent to $x$ and $x$ considers them the best, then there will be a unique lexicographically popular matching, namely $\mu =\{ (u_1,v_1),(u_2,v_2),(u_3,v_1),(u_4,v_2)\} \cup \{ (x,d_i)\mid i\in [q] \}$.

First we show that there can be no other lexicographically popular matching. By the same reasoning as before, $(u_2,v_2),(u_3,v_1)$ must be in $\mu $. Also, both $u_1$ and $u_4$ has to be matched. Therefore, the only other possibility for a lexicographically popular matching is the matching $\{ (u_1,v_2),(u_2,v_2),(u_3,v_1),(u_4,v_1)\}\cup \{ (x,d_i)\mid i\in [q] \}$, but it is dominated by $\{ (u_1,v_1),(u_2,v_1),(u_3,v_2),(u_4,v_2)\}\cup \{ (x,d_i)\mid i\in [q] \}$. 

Next we show that $\mu$ is lexicographically popular. All agents other than $u_2,u_3,v_1,v_2$ are saturated and matched to their best partners, hence only these four can improve. As $v$ is maximal, for any matching $\mu '$ to dominate $\mu$ there has to be an agent not from $\{u_2,u_3,v_1,v_2\}$ who gets worse, so the difference between the number of improving agents and the number of agents getting worse among $u_2,u_3,v_1,v_2$ must be at least 2. Let $\mu '$ be a matching that dominates $\mu$. Agents $u_2$ or $u_3$ could only improve if $v_1$ or $v_2$ gets worse respectively. So, $v_1$ and $v_2$ must both improve, while $u_2$ and $u_3$ must not be worse off. This is only possible if $v_1$ gets $u_4$ and $v_2$ gets $u_1$. But then $u_1$ and $u_4$ both get worse, so $\vote_{\rm lex} (\mu ,\mu ')\ge 0$, a contradiction.

\end{example}

With a counterexample in hand, we can show that deciding whether a lexicographically popular $b$-matching exists and verifying whether a $b$-matching is lexicographically popular are both hard. 
The proof uses a reduction from 
the NP-complete Exact 3-Cover problem (\textsc{x3c}). Given an instance $I$ of \textsc{x3c}, we can construct an instance $I'$ of the $b$-matching problem such that $I'$ has a unique candidate for lexicographically popular $b$-matching, and this candidate is lexicographically popular if and only if instance $I$ does not have an exact 3-cover.

\begin{figure}
    \centering
    \includegraphics[height=0.35\textheight]{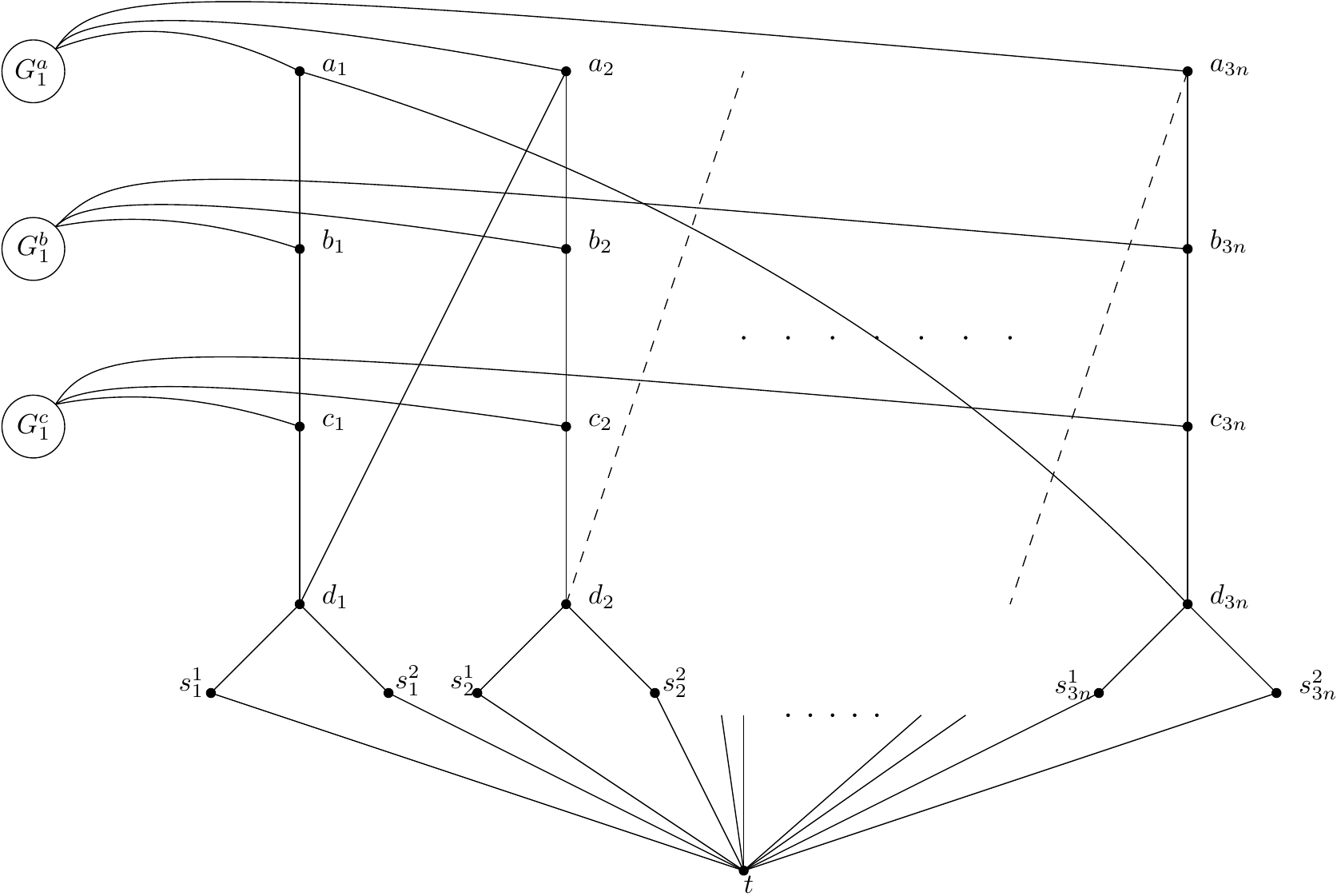}
    \caption{An illustration of the base of the construction of Theorem \ref{thm:lex_hardness}, with the 3 corresponding gadgets of a set $S_1=\{ 1,2,3n\}$. }
    \label{fig:lex-pop-constr}
\end{figure}
\begin{theorem} \label{thm:lex_hardness}
It is coNP-hard to decide if a given instance $(G;\succ ;b)$ admits a lexicographically popular $b$-matching. It is also coNP-complete to verify whether a given $b$-matching $\mu$ is lexicographically popular. These hold even if each agent has capacity at most 3.
\end{theorem}
\begin{proof}
We reduce from the NP-complete Exact 3-Cover problem (\textsc{x3c}). An instance $I$ of \textsc{x3c} consists of a set of elements $\mathcal{X}=\{ 1,2,\dots, 3n\}$ and a family of 3-element subsets of $\mathcal{X}$, $\mathcal{S}=\{ S_1,\dots ,S_{3n}\}$. The question is whether there exists a subset $S\subset \mathcal{S}$ such that each element is contained in exactly one member of $S$. We use the fact that \textsc{x3c} remains NP-hard even if each element appears in exactly three sets, hence the number of sets is also $3n$.

We create an instance $I'$ of our problem as follows.
\begin{itemize}
    \item [--] For each element $i\in [3n]$, we create 6 agents $a_i,b_i,c_i,d_i,s_i^1,s_i^2$ with capacities $b(a_i)=b(b_i)=b(c_i)=3$, $b(d_i)=2$, $b(s_i^1)=b(s_i^2)=1$.
    \item [--] For each set $S_j\in \mathcal{S}$ we create three gadgets $G_j^a$, $G_j^b$, $G_j^c$. Each is just a copy of Example~\ref{ex:no-instance}. For each gadget $G_j^l$ ($l\in \{ a,b,c\}$), we assign the corresponding $x$ agent, called $x_j^l$, capacity~3. $x_j^a$ is connected to the three $a_i$ agents corresponding to the elements in $S_j$, $x_j^b$ is connected to the 3 corresponding $b_i$, and $x_j^c$ is connected to the 3 corresponding $c_i$ agents.
    \item [--] We add a special agent $t$ with $b(t)=3$. 
    
\end{itemize}

We describe the preferences (which define the edges of the graph too).

\begin{center}
\begin{tabular}{ll|ll}
$a_i:$ & $b_i \succ x^a_{i_1} \succ x^a_{i_2} \succ x^a_{i_3} \succ d_{i-1}$ & $b_i:$ &  $c_i \succ x^b_{i_1} \succ x^b_{i_2} \succ x^b_{i_3} \succ a_i$\\
$c_i:$  &  $d_i \succ x^c_{i_1} \succ x^c_{i_2} \succ x^c_{i_3} \succ b_i$ & $d_i:$ & $a_{i+1} \succ s_i^1 \succ s_i^2 \succ c_i$\\
$s_i^{\ell}:$ & $d_i \succ t$ & $t:$ & $s_1^1 \succ s_1^2 \succ s_2^1 \succ s_2^2 \succ \dots \succ s_{3n}^1 \succ s_{3n}^2$,\\
\end{tabular}
\end{center}
where $i\in [3n]$, $\ell \in [2]$ and for an element $i\in \mathcal{X}$, $i_1\le i_2 \le i_3$ are the indices of the 3 sets containing $i$.
The preferences of the agents in the $G_j^l$ gadgets are inherited from Example \ref{ex:no-instance}, with the addition that for $l\in \{ a,b,c\}$, $j\in [3n]$ we add the three $l_i$ agents corresponding to the elements in $S_j$ to the beginning of $x_j^l$'s preference list (so they are considered best). The construction is illustrated in Figure \ref{fig:lex-pop-constr}.

First we show that there is only one possible candidate to be a popular $b$-matching in $I'$. By the observations in Example \ref{ex:no-instance}, we get that all $a_i,b_i,c_i$ agents must be saturated by being matched to all three of their $x_j^l$ neighbours, as otherwise there would be a matching $\mu '$ that only differs from $\mu$ on a gadget $G_j^l$ (where $x_j^l$ is not matched to all three of their best partners) that dominates $\mu .$ So, the agents $d_i$, $i\in [3n]$ can only be matched to their $s_i^1,s_i^2$ partners. Next, observe that $t$ cannot get any of the $s_i^{\ell}$ agents, as otherwise the corresponding $d_i$ would be unsaturated in $\mu$ and $s_i^{\ell}$ could switch to $d_i$, improving both of them and only worsening $t$. As $\mu$ must also be maximal, all $(d_i,s_i^1),(d_i,s_i^2)$ edges must be in $\mu$. Finally, by the observation in Example \ref{ex:no-instance}, the restriction of $\mu$ to a gadget $G_j^l$ must be the unique popular matching inside.

Therefore, we have shown that $I'$ admits a popular matching if and only if the above described $b$-matching $\mu$
 is popular. Hence, by showing that there is a matching $\mu '$ that dominates $\mu$ if and only if there is an exact 3-cover, we prove the hardness of both the existence and verification problems simultaneously.

\begin{claim}
If there is an exact 3-cover, then $\mu$ is not popular.
\end{claim}
\begin{proof}
Let the 3-cover be $S_{j_1},\dots,S_{j_n}$. Let $\mu '$ be the following $b$-matching: $\mu '$ is the same as $\mu$ on the graphs induced by the gadgets $G_j^l$. We add to $\mu '$ all edges of the form $(a_i,b_i)$, $(b_i,c_i)$, $(c_i,d_i)$, $(d_i,a_{i+1})$, $i+1$ taken modulo $3n$. Then, we add the edge $(t,s^1_1)$ to $\mu '$. Finally, for each of $a_i/b_i/c_i$ agent, we add to $\mu '$ the edge going to $x_j^a/x_j^b/x_j^c$ agent which corresponds to the index of the set that covers $i$ in the exact 3-cover. 

So, all $a_i,b_i,c_i,d_i$ agents and $t$ improve and all $s_i^1,s_i^2$ agents get worse. Finally, only $2n+2n+2n$ $x_j^l$ agents get worse. Hence, $\vote_{\rm lex}(\mu ',\mu )=3n+3n+3n+3n+1-6n-2n-2n-2n=1$, so $\mu '$ dominates $\mu$.
\end{proof}

\begin{claim}
\label{claim:lex-pop1}
For any matching $\mu '$ and any gadget $G_j^l$, it holds that if $\mu '$ does not contain all 3 of $x_j^l$'s best edges, then $\sum_{v\in G_j^l}\vote_{\rm lex}^v(\mu ',\mu )\le -1$.
\end{claim} 
\begin{proof}
As $x$ lost one of the best partners, they will vote with $-1$ anyway.
Since $\mu$ restricted to $G_j^l$ is popular, if there is a matching $\mu '$ with $\sum_{v\in G_j^l}\vote_{\rm lex}^v(\mu ',\mu )\ge 0$, then someone must improve by getting $x$. But only one agent, $v_2$ can improve their position by getting $x$, and can only do this by dropping only $u_4$ and keeping $u_2$, so $u_4$ must be worse off then, as $v_2$ was their best option. There is only one agent, $v_1$, who could improve now that $u_4$ is free, but only by keeping $u_3$   and leaving $u_1$, who must be worse off then. Now that $u_1$ is free, $v_2$ could improve even more, but it does not add another $+1$ to the vote. Hence, the number of worsening agents always remains more than the number of improving ones, so $\sum_{v\in G_j^l}\vote_{\rm lex}^v(\mu ',\mu )$ must be at most $-1$.
\end{proof}

\begin{claim}
If there is a matching dominating $\mu$, then there is an exact 3-cover.
\end{claim}
\begin{proof}
Suppose there is a matching $\mu '$ that dominates $\mu$. By Claim \ref{claim:lex-pop1} and the fact that $\mu$ is lexicographically popular inside the gadgets, we get that $\sum_{v\in G_j^l}\vote_{\rm lex}^v(\mu ',\mu )\le 0$, and $\sum_{v\in G_j^l}\vote_{\rm lex}^v(\mu ',\mu )\le -1$ whenever $x_j^l$ loses an edge. Also, all $s_i^{\ell}$ agents have their best partner, so they cannot improve. Therefore, only the $a_i,b_i,c_i,d_i$ agents and $t$ can improve. 

First suppose that not all $a_i,b_i,c_i,d_i$ agents improve. Let $k,\ell,m,p$ denote the number of improving $a_i,b_i,c_i$ and $d_i$ agents respectively. 
Let $C$ be the cycle consisting of edges of the form $(a_i,b_i),(b_i,c_i),(c_i,d_i),(d_i,a_{i+1})$ and let $c$ be the number of components of $\mu '$ restricted to $C$ with at least one edge.
Then, $\vote_{\rm lex}(\mu ',\mu )\le k+\ell +m+p+1-\frac{k+p}{3}-\frac{k+\ell}{3}-\frac{\ell +m}{3}-(m+p)-c$. 
This is because, even if $t$ improves, the number of agents who improve is at most $k+\ell +m+p+1$, while the $a_i$ agents must drop at least $\frac{k+p}{3}$ of their $x_j^l$ neighbours together, the $b_i$ agents must drop $\frac{k+\ell}{3}$ of their neighbours together, the $c_i$ agents have to drop $\frac{k+\ell}{3}$ and finally the $d_i$ agents must lose $m+p$ of their $s_i^1,s_i^2$ partners. Also, the last vertex of each component of size at least 2 (restricted to the cycle $C$) must also get worse, because they do not get their best partner, but must lose some partners to be able to be matched to their worst one (the previous in the cycle). It is easy to see that in each such component, the number of improving $a_i,b_i$ agents is at most two more than the number of improving $c_i,d_i$ agents. Also, if any of $a_i,b_i,c_i,d_i$ improves in $\mu '$, then they get the next agent in the cycle, hence all $k+\ell +m+p$ improving agents are inside components with at least one edge. Hence, $\ell +k\le m+p+2c$. Putting it all together, we get that $\vote_{\rm lex}(\mu ',\mu )\le \frac{k+\ell}{3}-\frac{m+p}{3}+1-c\le 1-\frac{c}{3}<1$. As $\vote_{\rm lex}(\mu ',\mu )$ is integer, we get that $\vote_{\rm lex}(\mu ',\mu )\le 0$, contradiction.

Now suppose that all $a_i,b_i,c_i,d_i$ agents improve in $\mu '$, so all edges of the cycle $C$ must be in $\mu '$. As only $t$ can improve and all $6n$ $s^\ell_i$ agents get worse, the sum of $\vote_{\rm lex}^v(\mu ',\mu )$ outside the gadgets is at most $12n+1-6n=6n+1$. Each $a_i,b_i,c_i,d_i$ must drop at least 2 original edges. Hence, the number of gadgets whose $x^l_j$ is not matched to the best three edges is at least $\frac{6n}{3}+\frac{6n}{3}+\frac{6n}{3}=6n$. As $\mu '$ dominates $\mu $ and we have Claim~\ref{claim:lex-pop1}, this equality should hold. Equality is only possible if each $a_i$ agent is able to keep one $x_j^l$ partner, such that there is $n$ $x_j^l$ agents who get to keep all their partners from $\mu $. This means that the corresponding sets must form an exact 3-cover.
\end{proof}

Theorem \ref{thm:lex_hardness} follows from the above claims.
\end{proof}

\paragraph*{Proportional voting} One might argue that agents should have voting weights proportional to their capacities in order to make the voting more fair. However, we can show that both the existence and verification problems remain hard even if all capacities are the same, using the following lemma. 

\begin{lemma}\label{lem:proportional}
For any instance $I=(G;\succ ;b)$ with maximum capacity $q$, we can create an instance $I'$, where every capacity is $q$, and there is a lexicographically popular $b$-matching in $I'$ if and only if there is one in $I$. Furthermore, a $b$-matching $\mu $ is lexicographically popular in $I$, if and only if by adding some fixed edges, the obtained $b$-matching $\mu '$ is lexicographically popular in $I'$.
\end{lemma}
\begin{proof}
If every capacity in $I$ is $q$, then $I'=I$ suffices.
Otherwise, for each agent $v$ with capacity $c<q$, we create $q-c$ dummy agents $d_1^v,\dots,d_{q-c}^v$ with capacity $q$, who only find $v$ acceptable, and $v$ considers them as best.

If a $b$-matching $\mu '$ in $I'$ does not contain all $(v,d_i^v)$ edges, then $\mu '$ cannot be lexicographically popular. Indeed, if $v$ drops the worst edge if needed and goes to $d_i^v$, then 2 agents improve and 1 gets worse, so we get a $b$-matching that dominates $\mu '$.

In $I'$, if a $b$-matching $\mu '$ containing all $(v,d_i^v)$ edges is not lexicographically popular, then there is a $b$-matching that includes all $(v,d_i^v)$ edges that dominates $\mu '$. This is because otherwise, if the dominating $b$-matching $\mu ''$ does not include an edge $(v,d_i^v)$, then deleting the worst edge from $v$ if $v$ was saturated and adding $(v,d_i^v)$, we get a $b$-matching, where the situation of $v$ improves (but can still be worse off than in $\mu '$), the vote of $d_i^v$ changes from $-1$ to $1$ and at most one agent's vote becomes worse, so it still dominates $\mu '$. By iterating this, we can obtain such a $b$-matching that dominates $\mu '$ too.

Therefore, if there is a lexicographically popular $b$-matching $\mu $ in $I$, then by adding to $\mu$ all $(v,d_i^v)$ edges, the obtained $b$-matching $\mu '$ is lexicographically popular in $I'$, because if it is not, then there is a $b$-matching $\mu ''$ containing all $(v,d_i^v)$ edges that dominates it. But, deleting the $(v,d_i^v)$ edges from $\mu ''$, we get a $b$-matching $\mu '''$ that dominates $\mu $ in $I$, contradiction.

If there is a lexicographically popular $b$-matching $\mu '$ in $I'$, then it contains all $(v,d_i^v)$ edges and by dropping them, we get a $b$-matching $\mu$ that is lexicographically popular in $I$. This is because if there was a $b$-matching dominating $\mu$, then adding all $(v,d_i^v)$ edges to it we would get a matching that dominates $\mu '$, contradiction.
\end{proof}

\section{Conclusion}
In this paper, we have shown that algorithmic results on two-sided many-to-many popular matchings can be generalized to two-sided popular matching problems with matroid constraints involving arbitrary matroids. The main tool that allows the extension to arbitrary matroids (not just weakly base orderable ones) is a new result on exchanges between two bases $I$ and $J$, that establishes a previously unknown relation between perfect matchings of exchanges from $I$ to $J$ and perfect matchings of exchanges from $J$ to $I$. This seems to be an interesting matroid property that may have applications in other areas involving ordered matroids, too.

Although the definition of popularity for matroid intersection is somewhat cumbersome, our results show that there is always a special maximum size popular solution that satisfies the stronger and more elegant property of \emph{super popularity}, and has maximum size among all solutions satisfying a weaker and more elegant property, \emph{weak defendability}. This is a remarkable phenomenon that shows a certain robustness of the definition of popularity.

We have also considered a different kind of natural voting method for many-to-many matchings with preferences, which results in the notion of lexicographic popularity. We have shown that, in contrast to popularity, both the existence problem and the verification problem are hard to decide. This raises the question whether the verification problem is efficiently solvable for popularity (or for the other notions that we have introduced: super-popularity, defendability, weak defendability). We could not settle these questions, so we pose it as an open problem: can we decide in polynomial time whether a given common independent set is popular?

\section*{Acknowledgement}

The work was supported by the Lend\"ulet Programme of the Hungarian Academy of Sciences -- grant number LP2021-1/2021, by the Hungarian National Research, Development and Innovation Office -- NKFIH, grant number TKP2020-NKA-06, and by JST PRESTO Grant Number JPMJPR212B.

\bibliographystyle{plain}
\bibliography{cit}

\begin{thebibliography}{10}

\bibitem{biro2022strong}
P{\'e}ter Bir{\'o} and Gergely Cs{\'a}ji.
\newblock Strong core and {P}areto-optimal solutions for the multiple partners
  matching problem under lexicographic preferences.
\newblock {\em arXiv preprint arXiv:2202.05484}, 2022.

\bibitem{brandl2018popular}
Florian Brandl and Telikepalli Kavitha.
\newblock Popular matchings with multiple partners.
\newblock In {\em 37th IARCS Annual Conference on Foundations of Software
  Technology and Theoretical Computer Science}, 2018.

\bibitem{brandl2019two}
Florian Brandl and Telikepalli Kavitha.
\newblock Two problems in max-size popular matchings.
\newblock {\em Algorithmica}, 81(7):2738--2764, 2019.

\bibitem{cechlarova2014pareto}
Katar{\'\i}na Cechl{\'a}rov{\'a}, Pavlos Eirinakis, Tam{\'a}s Fleiner,
  Dimitrios Magos, Ioannis Mourtos, and Eva Potpinkov{\'a}.
\newblock Pareto optimality in many-to-many matching problems.
\newblock {\em Discrete Optimization}, 14:160--169, 2014.

\bibitem{nicolas1785essai}
le~Marquis~de Condorcet, Marie Jean Antoine Nicolas de~Caritat.
\newblock {\em Essai sur l'application de l'analyse {\`a} la probabilit{\'e}
  des d{\'e}cisions rendues {\`a} la pluralit{\'e} des voix}.
\newblock de l'Imprimerie Royale, 1785.

\bibitem{csaji2022approximation}
Gergely Cs{\'a}ji, Tam{\'a}s Kir{\'a}ly, and Yu~Yokoi.
\newblock Approximation algorithms for matroidal and cardinal generalizations
  of stable matching.
\newblock In {\em Symposium on Simplicity in Algorithms (SOSA)}, pages
  103--113. SIAM, 2023.

\bibitem{faenza2019popular}
Yuri Faenza, Telikepalli Kavitha, Vladlena Powers, and Xingyu Zhang.
\newblock Popular matchings and limits to tractability.
\newblock In {\em Proceedings of the Thirtieth Annual ACM-SIAM Symposium on
  Discrete Algorithms}, pages 2790--2809. SIAM, 2019.

\bibitem{fleiner2001matroid}
Tam{\'a}s Fleiner.
\newblock A matroid generalization of the stable matching polytope.
\newblock In {\em Proc. 8th International Conference on Integer Programming and
  Combinatorial Optimization}, pages 105--114. Springer, 2001.

\bibitem{fleiner2003fixed}
Tam{\'a}s Fleiner.
\newblock A fixed-point approach to stable matchings and some applications.
\newblock {\em Mathematics of Operations research}, 28(1):103--126, 2003.

\bibitem{GS62}
David Gale and Lloyd~S Shapley.
\newblock College admissions and the stability of marriage.
\newblock {\em American Mathematical Monthly}, 69(1):9--15, 1962.

\bibitem{gardenfors1975match}
Peter G{\"a}rdenfors.
\newblock Match making: assignments based on bilateral preferences.
\newblock {\em Behavioral Science}, 20(3):166--173, 1975.

\bibitem{gupta2019popular}
Sushmita Gupta, Pranabendu Misra, Saket Saurabh, and Meirav Zehavi.
\newblock Popular matching in roommates setting is {NP}-hard.
\newblock In {\em Proceedings of the Thirtieth Annual ACM-SIAM Symposium on
  Discrete Algorithms}, pages 2810--2822. SIAM, 2019.

\bibitem{gupta2021popular}
Sushmita Gupta, Pranabendu Misra, Saket Saurabh, and Meirav Zehavi.
\newblock Popular matching in roommates setting is {NP}-hard.
\newblock {\em ACM Transactions on Computation Theory (TOCT)}, 13(2):1--20,
  2021.

\bibitem{huang2011popular}
Chien-Chung Huang and Telikepalli Kavitha.
\newblock Popular matchings in the stable marriage problem.
\newblock In {\em International Colloquium on Automata, Languages, and
  Programming (ICALP)}, pages 666--677. Springer, 2011.

\bibitem{kamiyama2020popular}
Naoyuki Kamiyama.
\newblock Popular matchings with two-sided preference lists and matroid
  constraints.
\newblock {\em Theoretical Computer Science}, 809:265--276, 2020.

\bibitem{kavitha2014size}
Telikepalli Kavitha.
\newblock A size-popularity tradeoff in the stable marriage problem.
\newblock {\em SIAM Journal on Computing}, 43(1):52--71, 2014.

\bibitem{nasre2017popularity}
Meghana Nasre and Amit Rawat.
\newblock Popularity in the generalized hospital residents setting.
\newblock In {\em International Computer Science Symposium in Russia}, pages
  245--259. Springer, 2017.

\bibitem{paluch2014popular}
Katarzyna Paluch.
\newblock Popular and clan-popular $b$-matchings.
\newblock {\em Theoretical Computer Science}, 544:3--13, 2014.

\bibitem{schrijver2003combinatorial}
Alexander Schrijver.
\newblock {\em Combinatorial Optimization: Polyhedra and Efficiency},
  volume~24.
\newblock Springer, 2003.

\end{thebibliography}
 
\end{document}